\documentclass[conference,letterpaper]{IEEEtran}

\usepackage[utf8]{inputenc} 
\usepackage[T1]{fontenc}
\usepackage{url}
\usepackage{ifthen}
\usepackage{cite}
\usepackage[cmex10]{amsmath} % Use the [cmex10] option to ensure complicance
\usepackage{amssymb}
\usepackage{amsfonts}
\usepackage{mathabx}
\usepackage{array}
\usepackage{enumitem}
\usepackage{comment}
\usepackage{xcolor}
\usepackage{amsthm}
\usepackage{graphicx}
\usepackage{placeins}
\usepackage{subcaption}

\newtheorem{theorem}{Theorem}

\definecolor{newred}{HTML}{ff382e}

\def\E{\mathbb{E}}
\title{Estimators for Substitution Rates in Genomes \\ from Read Data}
\author{%
  \IEEEauthorblockN{Shiv Pratap Singh Rathore \, and \, Navin Kashyap}
  \IEEEauthorblockA{Department of Electrical Communication Engineering \\
                   Indian Institute of Science, Bengaluru \\
                    Email: \{shivprataps,nkashyap\}@iisc.ac.in}
}
\date{}

\begin{document}

\maketitle

\begin{abstract}
We study the problem of estimating the mutation rate between two sequences from noisy sequencing reads. Existing alignment-free methods typically assume direct access to the full sequences. We extend these methods to the sequencing framework, where only noisy reads from the sequences are observed. We use a simple model in which both mutations and sequencing errors are substitutions. We propose multiple estimators, provide theoretical guarantees for one of them, and evaluate the others through simulations.
\end{abstract}

\section{Introduction}
Estimating mutation rates in genomes is a fundamental task in genomics, with applications in comparative genomics, phylogenetics, population genetics, evolutionary biology, and disease genomics, including understanding the dynamics of cancer mutations and viral evolution. Traditional alignment-based methods for estimating mutation rates can be both computationally intensive and prone to inaccuracies, especially in highly repetitive or structurally complex regions of the genome. To address these limitations, alignment-free approaches based on $k$-mer statistics (a $k$-mer being a substring of a fixed length $k$) have gained attention in recent years. In this study, we focus on a simplified problem: estimating mutation rates between two genomes under the assumption that the only mutations are substitution errors following an independent and identically distributed (i.i.d.) model.

This problem has previously been studied under the assumption that the genome contains no repeats \cite{Ondov029827}, \cite{Blanca2021.01.15.426881}. This assumption was later relaxed in \cite{Wu2025.06.19.660607}, which allows repeated $k$-mers but assumes that whenever a $k$-mer mutates, it does not mutate to any other $k$-mer present in the string. The work in \cite{Ondov029827} estimates the substitution rate using the Jaccard similarity between $k$-mer sketches of two strings, while \cite{Blanca2021.01.15.426881} and \cite{Wu2025.06.19.660607} adopt a method-of-moments approach based on the number of distinct common $k$-mers between the source and mutated strings. The former does not account for $k$-mer multiplicities and therefore performs poorly for highly repetitive sequences, whereas the latter remains empirically effective even in the presence of high repeat content. A common limitation of all these approaches is that they require access to the abundance histogram of the source string. In contrast, we extend these ideas and introduce new techniques for the more realistic case in which only unassembled sequencing reads are available, without access to the underlying sequence or its abundance histogram. Here, a sequencing read is a substring of the underlying sequence, possibly corrupted by sequencing errors that occur during the physical process of acquiring the read; however, the location from which a given read originates is unknown.

The rest of the paper is organized as follows. Section~\ref{sec:prelims} contains some definitions and notations used throughout the paper; it also introduces the problem statement. Section~\ref{sec:nonseq} presents, for benchmarking purposes, estimators for the scenario when there is direct access to $k$-mer data from the genomic sequences. In Section~\ref{sec:seq}, we derive estimators for the situation where only noisy sequencing reads are available. Section~\ref{sec:analysis} provides some performance analysis, while Section~\ref{sec:sim} shows simulation results. The paper concludes in Section~\ref{sec:conc} with a summary and discussion of possible future research directions.

\section{Notation and Preliminaries} \label{sec:prelims}
Genomic sequences\footnote{In this paper, the terms \emph{string} and \emph{sequence} are used interchangeably. We distinguish between \emph{substrings}, which are contiguous, and \emph{subsequences}, which may be non-contiguous.} are composed of nucleotides, represented by the four letters $\mathrm{A},\mathrm{C},\mathrm{G},\mathrm{T}$. Let $x$ be a circular\footnote{The choice of a \emph{circular} genome is purely for convenience, so that we avoid having to deal with edge effects. In any case, for long strings, the edge effects will be negligible.} genomic sequence of length $G$. Fixing an arbitrary starting position, we use $x_i$ to denote the nucleotide at the $i$-th position of $x$, i.e., $x=\{x_1,\cdots,x_G \}$. We use $x[i:j]$ to denote the $(j-i+1)$-length substring $x[i:j]:=\{x_i,x_{i+1},\cdots,x_j \}$. Let $k$ be the parameter indicating $k$-mer size where a $k$-mer is a substring of length $k$, i.e., $x[i:i+k-1], 1 \leq i \leq G$. Thus, sequence $x$ contains exactly $G$ (possibly repeating) $k$-mers since $x$ is circular. (For a non-circular string, the number of $k$-mers would be $G-k+1$.) Let $\mathcal{K}$ denote the set of all distinct $k$-mers in $x$. For a $k$-mer $v \in \mathcal{K}$, we use $f_v$ to denote the number of occurences of $k$-mer $v$ in $x$. For any two $k$-mers $v,w \in \mathcal{K}$, we use $d_{H}(v,w)$ to denote Hamming distance between the two $k$-mers. The abundance histogram of $x$ is the sequence $(a_1,\cdots,a_G)$ where $a_i$ is the number of $k$-mers in $\mathcal{K}$ that occur $i$ times in $x$, i.e., $a_i = |\{v \in \mathcal{K}: f_v = i\}|$.

Let $y$ be the string (also circular) obtained by passing the circular string $x$ through a random substitution process with a rate parameter $p$: independently, the character at each position in $x$ is unchanged with probability $1-p$ and changed to one of three other nucleotides with probability $p/3$ per nucleotide. We use $f_v'$ to denote the number of occurrences of $k$-mer $v$ in $y$, and $\mathcal{K}'$ to denote set of all distinct $k$-mers in $y$.

A noiseless read of $x$ is a random substring of $x$. We make the simplifying assumption that reads are of a fixed length $L$, and the starting position of a read is independently and uniformly distributed, i.e., $\Pr(\text{a read starts at $i$-th position)}=\frac{1}{G} \ \forall i \in {1,2,\cdots,G}$. Reads are assumed to be acquired through a noisy physical sequencing process, with sequencing errors modeled as i.i.d.\ substitutions. Specifically, a noisy read of $x$ is obtained by passing a random substring $x[i:i+L-1]$ through the random substitution process described above, but this time parameterized by the sequencing error rate $s$. Multiple independent noisy reads of $x$ are obtained in this manner. If $N$ is the number of reads so obtained, the sequencing coverage of the sequence $x$ is defined to be $c := NL/G$; this is the expected number of times a given position in $x$ is covered by the $N$ acquired reads. 

\subsection{Problem Statement}
Throughout, the source string $x$ is assumed to be arbitrary but fixed. The mutated string $y$ obtained through the substitution process is random. The problem we consider is the following: \emph{Given $N$ independent noisy reads acquired from each of the sequences $x$ and $y$ (so $2N$ reads in all), estimate the rate $p$ of the substitution process that generates $y$ from $x$.}
The sequences $x$ and $y$ are not known to the estimator. The sequencing error rate $s$ may be known or unknown, depending on the specific application.

\subsection{Concentration Inequalities}
In our analysis of estimator performance, we make use the following two well-known concentration inequalities.

\subsubsection{Hoeffding's Inequality \cite{hoeffding1963probability}} 
Let $X_1,\cdots,X_n$ be independent random variables such that $a_i \leq X_i \leq b_i$ almost surely. If $X=\sum\limits_{i=1}^{n} X_i$, then  
$$\Pr(|X-\E[X]| \geq t) \leq 2 \exp \biggl(-\frac{2t^2}{\sum\limits_{i=1}^{n}(b_i-a_i)^2} \biggr)$$

\subsubsection{McDiarmid's Inequality \cite{mcdiarmid1989method}} 
Let $X_1,\ldots,X_n$ be independent random variables taking values in sets
$\mathcal{X}_1,\ldots,\mathcal{X}_n$.  
Suppose $f:\mathcal{X}_1 \times \cdots \times \mathcal{X}_n \to \mathbb{R}$
satisfies the bounded differences condition : for every $i=1,\!\cdots\!,\!n$,
$\left|f(x_1,\ldots,x_i,\ldots,x_n)-f(x_1,\ldots,x_i',\ldots,x_n)\right| \!\leq\! c_i$.
Then, for all $t>0$, 
\begin{align*}\Pr\!\bigl(|f(X_1,\ldots,X_n) - \mathbb{E} & [f(X_1,\ldots,X_n)]| \geq t \bigr) \\ 
& \ \ \ \ \ \ \ \ \ \ \leq 2\exp\!\left(-\frac{2t^2}{\sum_{i=1}^n c_i^2} \right). \end{align*}

\section{Estimators with access to $k$-mer data from the two strings} \label{sec:nonseq}
In this section and the next, we introduce a number of $k$-mer-based estimators for $p$: some use $k=1$ but others use larger $k$. The performance of these estimators depends on the underlying source strings $x$ as well as on the mutation rate $p$. To benchmark their behavior, we first work directly with $k$-mer data from the strings $x$ and $y$. We give the estimators in this section access to the $k$-mer frequency tables, i.e., the number of occurrences of each $k$-mer in $x$ and $y$. In the next section, we extend these estimators to the more realistic case of sequencing data, where only unassembled reads from $x$ and $y$ are observed.

\subsection{$k=1$}
The idea of using $k=1$ is inspired by the work of Hera et al.\ \cite{Hera2025.05.14.653858} who use the number of occurrences of a given nucleotide, say $\mathrm{A}$, in the source string $x$ and the mutated string $y$ as one of the parameters to estimate substitution and indel rates. In our work, we restrict attention to the substitution-only case and use this statistic accordingly, before extending the same idea to any value of $k$. 

For 1-mer $\mathrm{A}$, $f_{\mathrm{A}}$ and $f'_{\mathrm{A}}$ represents the number of times $\mathrm{A}$ occurs in $x$ and $y$ respectively. Now, $f_{\mathrm{A}}'=\sum\limits_{i=1}^{G} Y_i$ where $Y_i$'s are indicator random variables such that $Y_i= 1$ iff $y_i=\mathrm{A}$. Now, if $x_i=\mathrm{A}$, then $\Pr(y_i=\mathrm{A})=1-p$, otherwise $\Pr(y_i=\mathrm{A})=p/3$. Thus, $\E[f_{\mathrm{A}}']=\sum\limits_{i:x_i=\mathrm{A}} (1-p) + \sum\limits_{i:x_i \neq A} \frac{p}{3} = f_{\mathrm{A}} (1-p)+(G-f_{\mathrm{A}})\frac{p}{3}$. Replacing $\E[f'_{\mathrm{A}}]$ with the observed value $f'_{\mathrm{A}}$, we get 
\begin{equation} \label{def:phat-1}
    \widehat{p}=\frac{3\cdot (f'_{\mathrm{A}}-f_{\mathrm{A}})}{G-4f_{\mathrm{A}}}
\end{equation}
We provide a performance analysis of this estimator in Section~\ref{sec:phat-1}.

We can use a subset of 1-mers as well, let's say, we take $\mathrm{G}$ and $\mathrm{C}$ 1-mers. Then, GC fraction can be used for estimation. Note that GC fraction is given as 
$\frac{f_\mathrm{G}+f_\mathrm{C}}{G}$. Let $x_{\mathrm{GC}}$ and $y_{\mathrm{GC}}$ represent GC fraction in $x$ and $y$ respectively.

Now, $\E[f'_{\mathrm{G}}+f'_{\mathrm{C}}]=(f_{\mathrm{G}}+f_{\mathrm{C}})(1-\frac{2p}{3})+(G-f_{\mathrm{G}}-f_{\mathrm{C}})\frac{2p}{3} \Longrightarrow \E[y_{\mathrm{GC}}]=x_{\mathrm{GC}}(1-\frac{2p}{3})+(1-x_{\mathrm{GC}})\frac{2p}{3}$. Thus,\footnote{For now, to keep notation simple, we use $\widehat{p}$ to denote all our estimators; the specific estimator in use is determined by the context.}
\begin{equation} \label{def:phat-2}
    \widehat{p}=\frac{3 \cdot (y_{\mathrm{GC}}-x_{\mathrm{GC}})}{2-4x_{\mathrm{GC}}}
\end{equation}

This idea can be extended to general values of $k$ as well.

\subsection{$k \ge 1$} 
Observe that for a $k$-mer $v$, we have $\E[f'_{v}]=\sum\limits_{w \in \mathcal{K}} f_{w}(1-p)^{k-d_{H}(v,w)}(p/3)^{d_{H}(v,w)}$. The idea is to consider a subset $\mathcal{S} \subset \mathcal{K}$ such that $\sum\limits_{v \in \mathcal{S}} f_{v}$ is large enough that $\sum\limits_{v \in \mathcal{S}} f'_{v}$ can be confidently assumed to be close to its expected value, $\sum\limits_{v \in \mathcal{S}} \E[f'_{v}]=\sum \limits_{v \in \mathcal{S}}\sum\limits_{w \in \mathcal{K}} f_{w}(1-p)^{k-d_{H}(v,w)}(p/3)^{d_{H}(v,w)}$. %\nknote{What is ``large enough'' and why is this needed?} (\textcolor{green}{Large enough so that we get enough information to get a good estimate but it's not important to mention in this case, but in the sequencing case for large $k$ estimator, it is important to mention as if information is less, estimator will perform poorly there.}) then $\sum\limits_{v \in \mathcal{S}} \E[f'_{v}]=\sum \limits_{v \in \mathcal{S}}\sum\limits_{w \in \mathcal{K}} f_{w}(1-p)^{k-d_{H}(v,w)}(p/3)^{d_{H}(v,w)}$. 
Then, the estimate $\widehat{p}$ can be obtained numerically as the solution of
\begin{equation} \label{def:phat-3}
    \sum\limits_{v \in \mathcal{S}} f'_{v}=\sum \limits_{v \in \mathcal{S}}\sum\limits_{w \in \mathcal{K}} f_{w}(1-\widehat{p})^{k-d_{H}(v,w)}\left(\frac{\widehat{p}}{3}\right)^{d_{H}(v,w)}.
\end{equation}

\subsection{Large $k: \ k \geq 20$}
For large values of $k$, $|\mathcal{K}|$ is quite large and $|\mathcal{K}| \cdot |\mathcal{S}|$ pairwise Hamming distances need to be computed, which is computationally expensive, and moreover, it makes the estimator quite complex. To address this issue, we adopt the simplifying assumption used in \cite{Wu2025.06.19.660607}, namely that if a $k$-mer $v$ mutates, it doesn't mutate to any other $k$-mer $w \in \mathcal{K}$ with $w \neq v$.

Under this assumption, $\E[f'_{v}]=f_v (1-p)^{k} \Longrightarrow \sum\limits_{v \in \mathcal{S}}\E[f'_{v}] = \sum\limits_{v \in \mathcal{S}} f_{v} (1-p)^{k} = (1-p)^{k} \sum\limits_{v \in \mathcal{S}} f_{v}$. If we take $\mathcal{S}=\mathcal{K}$, then $\sum\limits_{v \in \mathcal{K}}\E[f'_{v}]=G(1-p)^k $. We thus obtain the estimator
\begin{equation} \label{def:phat-4}
\widehat{p}=1-\left(\dfrac{\sum\limits_{v \in \mathcal{K}}f'_{v}}{G} \right)^{1/k}.
\end{equation}

Note that our estimator is similar to the estimators designed for this problem in \cite{Blanca2021.01.15.426881},\cite{Wu2025.06.19.660607}, the difference in our approach being that we account for multiplicities of $k$-mers in both $x$ and $y$.

Next, we allow our estimators to access only the noisy reads observed from sequencing data.

\section{Estimators with access only to noisy sequencing reads} \label{sec:seq}
 Now, the goal is to estimate the mutation rate $p$ given $N$ noisy reads each for $x$ and $y$ of fixed length $L$.

For a $k$-mer $v$, we use $h_v$ and $h'_{v}$ to denote number of occurrences of $k$-mer $v$ across all the $N$ reads of $x$ and $y$ respectively.

Now, consider the reads of string $x$. Then, $\Pr(\text{positions } x[i:i+k-1] \text{ are covered by a particular read})=\frac{L-k+1}{G} \ \forall i \in \{1,\cdots,G \}$. Let $r_i$ and $r^{'}_{i}$ denote number of times positions $x[i:i+k-1]$ and $y[i:i+k-1]$ are covered across all the reads of $x$ and $y$ respectively. Then, $\E[r_i]=\E[r^{'}_{i}]=\frac{N(L-k+1)}{G}$. 

Thus, for any $k$-mer $v$,
\begin{align}
    \E&[h_v] \notag \\ &=\sum\limits_{i=1}^{G}\!\E[r_i] (1\!-s)^{k-d_{H}(v,x[i:i+k-1])}\biggl(\frac{s}{3}\biggr)^{d_{H}(v,x[i:i+k-1])}  \nonumber \\
    &= \sum\limits_{w \in \mathcal{K}}f_w \!\left(\!\frac{N(L-k+1)}{G} \!\right) \!(1-s)^{k-d_{H}(v,w)}\biggl(\frac{s}{3}\biggr)^{d_{H}(v,w)} \label{eq:reads_kmer}
\end{align}

\subsection{$k=1$}
For $k=1$, we use fraction of $\mathrm{A}$'s in reads of $x$ and $y$, i.e., $\frac{h_{\mathrm{A}}}{NL}$ and $\frac{h'_{\mathrm{A}}}{NL}$ to estimate $p$. From (\ref{eq:reads_kmer}), we get
\begin{align*}
    \E[h_{\mathrm{A}}] &=\frac{NL}{G} \left(f_{\mathrm{A}}(1-s)+(G-f_{\mathrm{A}})\frac{s}{3} \right) \\
    &=\frac{NL}{G} \left(f_{\mathrm{A}} \left(1-\frac{4s}{3} \right)+G \frac{s}{3} \right) 
\end{align*}
If $f'_{\mathrm{A}}$ is known, then analogously we get $$\E[h'_{\mathrm{A}}|f'_{\mathrm{A}}]=\frac{NL}{G} \left(f'_{\mathrm{A}} \left(1-\frac{4s}{3} \right)+G \frac{s}{3} \right)$$ Then,
\begin{align*}
  \E[h'_{\mathrm{A}}]&= \E[\E[h'_{\mathrm{A}} \mid f'_{\mathrm{A}}]]  \\
  &= \frac{NL}{G} \left(\E[f'_{\mathrm{A}}]\left(1-\frac{4s}{3} \right)+G\frac{s}{3} \right)  \\ 
  &=\frac{NL}{G} \left(\bigg(f_{\mathrm{A}}\bigg(1-\frac{4p}{3} \bigg)+G \frac{p}{3} \bigg) \bigg(1-\frac{4s}{3}\bigg)+G\frac{s}{3} \right)  \\
  & =\frac{NL}{G} \left(\bigg(f_{\mathrm{A}} \bigg(1-\frac{4s}{3} \bigg)+G \frac{s}{3} \bigg) \bigg(1-\frac{4p}{3} \bigg)+G\frac{p}{3} \right) \\
  & = \E[h_{\mathrm{A}}] \bigg(1-\frac{4p}{3} \bigg)+NL\frac{p}{3} 
\end{align*}

Thus, 
\begin{equation} \label{def:phat-5}
    \widehat{p}=\dfrac{3(h'_{\mathrm{A}}-h_{A})}{NL-4h_{\mathrm{A}}}
\end{equation}
Note that our estimator is independent of $s$. An analysis of the performance of this estimator is provided in Section~\ref{sec:phat-5}.

\subsection{Large $k$: $k \geq 20$}
In this case, we assume the following: 
\begin{enumerate}
    \item[(1)] If a $k$-mer $v$ mutates, it doesn't mutate to any other $k$-mer $w \in \mathcal{K}$, with $w \neq v$.
    \item[(2)] Due to sequencing errors, a $k$-mer $v$ in the string $x$ (or $y$) doesn't change to any other $k$-mer $w \in \mathcal{K}$ (or $\mathcal{K}^{'}$), with $w \neq v$, in the reads of $x$ (or $y$). 
\end{enumerate}
Then, from (\ref{eq:reads_kmer}), we obtain $\E[h_{v}]=f_v \cdot  \bigg(\dfrac{N(L-k+1)}{G} \bigg) (1-s)^k $. Analogously, we get $\E[h'_{v}|f'_{v}]= f'_v \cdot  \bigg(\dfrac{N(L-k+1)}{G} \bigg) (1-s)^k$. Using similar arguments and calculations as before, we get $\E[h'_{v}]= \E[h_v] (1-p)^k$. For a subset $\mathcal{S} \subseteq \mathcal{K}$, $\sum\limits_{v \in \mathcal{S}} \E[h'_{v}]=(1-p)^k \sum\limits_{v \in \mathcal{S}} \E[h_v] $. Thus,

\vspace{-2mm}
\begin{equation}\label{def:phat-6}
\widehat{p}=1-\left(\dfrac{\sum\limits_{v \in \mathcal{S}} h'_{v}}{\sum\limits_{v \in \mathcal{S}} h_v} \right)^{1/k}
\end{equation}
\vspace{-2mm}

Ideally, we would like to take $\mathcal{S}=\mathcal{K}$. However, in practice, some $k$-mers present in the string may not appear in the reads at all, while additional $k$-mers may be introduced due to sequencing errors. Let $\mathcal{R}$ and $\mathcal{R}'$ denote the sets of all $k$-mers observed in the reads of $x$ and $y$, respectively. We define $\mathcal{F}=\mathcal{R}\setminus\mathcal{K}$ and $\mathcal{F}'=\mathcal{R}'\setminus\mathcal{K}'$ as the sets of spurious $k$-mers introduced by sequencing errors in the reads of $x$ and $y$.

Given a $k$-mer $v\in\mathcal{R}$, it is not possible to determine whether $v\in\mathcal{K}$ or $v\in\mathcal{F}$. Therefore, we use the following estimator:
\begin{equation}\label{def:phat-8}
\widehat{p}
=
1-\left(
\frac{\sum\limits_{v\in\mathcal{R}:h_v\geq \lambda} h'_v}
{\sum\limits_{v\in\mathcal{R}:h_v\geq \lambda} h_v}
\right)^{1/k}.
\end{equation}

Observe that
$$
\sum_{v\in\mathcal{R}:h_v\geq \lambda} h_v
=
\sum_{v\in\mathcal{K}:h_v\geq \lambda} h_v
+
\sum_{v\in\mathcal{F}:h_v\geq \lambda} h_v.
$$
We choose the threshold $\lambda$ such that the total mass $\sum_{v\in\mathcal{R}:h_v\geq \lambda} h_v$ is sufficiently large, while the contribution from spurious $k$-mers, $\sum_{v\in\mathcal{F}:h_v\geq \lambda} h_v$, is minimized.

Since a $k$-mer in the reads remains unaffected by sequencing errors with probability $(1-s)^k$, we have
$\mathbb{E}[$fraction of $k$-mers in reads unaffected by sequencing errors $]=(1-s)^k$. Under Assumption~(2), this implies
$$
\mathbb{E}\!\left[
\frac{\sum\limits_{v\in\mathcal{K}} h_v}
{\sum\limits_{v\in\mathcal{R}} h_v}
\right]
=
(1-s)^k.
$$

Motivated by this observation, we select
$$
\lambda
=
\max _{\lambda' \geq 2}\left\{
\lambda' :
\sum_{v\in\mathcal{R}:h_v\geq \lambda'} h_v
\geq
(1-s)^k \sum_{v\in\mathcal{R}} h_v
\right\}.
$$
That is, we retain the most frequently occurring $k$-mers in the reads whose cumulative frequency accounts for approximately $(1-s)^k$ fraction of all observed $k$-mers. When the sequencing error rate $s$ is small, true $k$-mers $v\in\mathcal{K}$ typically appear more frequently in the reads than spurious $k$-mers $w\in\mathcal{F}$, resulting in only a small number of outliers and a negligible contribution from $\sum_{v\in\mathcal{F}:h_v\geq \lambda} h_v$.

A limitation of this estimator is that selecting the parameter $\lambda$ requires prior knowledge of the sequencing error rate $s$. Although this requirement can be relaxed by using a known upper bound on $s$, doing so leads to degraded performance. In contrast, the estimator for the $k=1$ case remains applicable even when $s$ is unknown. On the other hand, the sequence length $G$ is not required to be known for both the estimators.

\section{Performance Analysis} \label{sec:analysis}
We here analyze the performance of the $1$-mer-based estimators $\widehat{p}$ in \eqref{def:phat-1} and \eqref{def:phat-5} using the concentration inequalities of Hoeffding and McDiarmid to bound the probability (with respect to the random substitution process) that $|\widehat{p}-p| \leq \epsilon p$. We have not been successful in extending this style of analysis to the estimators that use large values of $k$, so we evaluate their performance using simulation results in the next section.

\subsection{A concentration bound for the estimator in \eqref{def:phat-1}} \label{sec:phat-1}
For the estimator $\widehat{p}$ given by \eqref{def:phat-1}, we have
\begin{align*}
    \Pr(|\widehat{p}-p| \leq \epsilon p) &=\Pr \left( \bigg|\frac{3\cdot (f'_{\mathrm{A}}-f_{\mathrm{A}})}{G-4f_{\mathrm{A}}}-p \bigg|\leq \epsilon p \right) \\
    &= \Pr \left( \bigg|f'_{\mathrm{A}}-\E[f'_{\mathrm{A}}] \bigg| \leq \frac{|G-4f_{\mathrm{A}}| \cdot p \epsilon}{3} \right) \\
    & \geq 1-2 \exp \bigg(-\frac{2|G-4f_{\mathrm{A}}|^2 \cdot p^2 \epsilon ^2}{9G} \bigg) \\
    &= 1 - 2 \exp \bigg(-\frac{32}{9}G\bigg|\frac{f_{\mathrm{A}}}{G}-0.25 \bigg|^2 p^2 \epsilon^2 \bigg)
\end{align*}
where the inequality above is obtained using Hoeffding's inequality.

Now, for fixed values of $\epsilon, G$ and $p$, we compute the fraction of $\mathrm{A}$'s required in the source string $x$ such that $\Pr(|\widehat{p}-p| \geq \epsilon p) \leq 2 \exp(-5.3) \leq 10^{-2}$. This requirement translates to  $\bigg|\frac{f_{\mathrm{A}}}{G}-0.25 \bigg| \geq \sqrt{\frac{1.5}{p^2 \epsilon^2 G}}$. For this estimator to be effective, $\bigg|\frac{f_{\mathrm{A}}}{G}-0.25 \bigg| > 0$ and the minimum deviation needed to guarantee the desired performance decreases as $p$ and $G$ increase, as illustrated in Table~\ref{tab:min_deviation_fA}.

\begin{table}[h!]
\centering
\caption{Minimum value of $ \big|\frac{f_{\mathrm{A}}}{G}-0.25\big|$ required to ensure performance guarantees for $\varepsilon=0.1$}
\label{tab:min_deviation_fA}
\begin{tabular}{c|cccc}
\hline
$\mathbf{p \backslash G}$ & $\mathbf{10^4}$ & $\mathbf{10^5}$ & $\mathbf{10^6}$ & $\mathbf{10^7}$ \\
\hline
$\mathbf{0.01}$ & 12.24  & 3.87  & 1.22  & 0.387 \\
$\mathbf{0.03}$ & 4.08 & 1.29 & 0.408 & 0.129 \\
$\mathbf{0.05}$ & 2.44 & 0.77 & 0.245 & 0.077 \\
$\mathbf{0.10}$ & 1.22 & 0.387 & 0.122 & 0.038 \\
$\mathbf{0.20}$ & 0.612 & 0.194 & 0.061 & 0.019 \\
$\mathbf{0.50}$ & 0.245 & 0.077 & 0.024 & 0.008 \\
\hline
\end{tabular}
\end{table}

\subsection{A concentration bound for the estimator in \eqref{def:phat-5}} \label{sec:phat-5}
\begin{theorem}
Assume $p,s<\frac{3}{4}$ 
%\nknote{Where is this assumption needed?} \textcolor{green}{If $s=0.75$ or $p=0.75$, then many things become unstable, for example, mutated sequence and sequencing output are completely random and if we allow $s \text{ or } p > 0.75$, for few steps, $(1-4s/3)$ and $(1-4p/3)$ need to be changed to $|1-4s/3|$ and $|1-4p/3|$ respectively. It can be treated in a similar way as BSC where we assume $p < 0.5$. If you want to discard this assumption, please feel free to do so.} 
and
$$
\left|\frac{f_{\mathrm{A}}}{G}-0.25\right|
\ge
\frac{3}{4p\epsilon}\left(\sqrt{\frac{C_2}{2N}}+\sqrt{\frac{C_1}{2G}}\right)
+\sqrt{\frac{C_3}{2N}}
+\left|\frac{4s}{3}\frac{f_{\mathrm{A}}}{G}-\frac{s}{3}\right|
$$
for some $C_1, C_2, C_3 > 0$. Then, for the estimator $\widehat{p}$ given by \eqref{def:phat-5}, we have
$$
\Pr(|\widehat{p}-p|\le \epsilon p)
\ge
1-2e^{-C_1}-2e^{-C_2}-2e^{-C_3}.
$$
\end{theorem}

\begin{proof}
    
\begin{align*}
    \Pr(|\widehat{p}-p| \leq \epsilon p) &=\Pr \left( \bigg|\frac{3\cdot (h'_{\mathrm{A}}-h_{\mathrm{A}})}{NL-4h_{\mathrm{A}}}-p \bigg|\leq \epsilon p \right) \\
    &= \Pr \left( \bigg|h'_{\mathrm{A}}-\E[h'_{\mathrm{A}}] \bigg| \leq \frac{|NL-4h_{\mathrm{A}}| \cdot p \epsilon}{3} \right) 
\end{align*}

Define the event $B_1:= \left\{|f'_{\mathrm{A}}-\E[f'_{\mathrm{A}}]| \leq \delta\cdot\E[f'_{\mathrm{A}}]\right\}$ for some suitable choice of $\delta$ to be specified later. Let $t=|NL-4h_{\mathrm{A}}| \cdot p\epsilon/3$. We bound $\Pr(|h'_{\mathrm{A}}-\E[h'_{\mathrm{A}}]| > t) = 1 - \Pr(|h'_{\mathrm{A}}-\E[h'_{\mathrm{A}}]| \le t)$ from above as follows:
% \begin{align*}
%     &\Pr(|h'_{\mathrm{A}}-\E[h'_{\mathrm{A}}]| \leq t) \\
%     & \quad \geq \Pr(|h'_{\mathrm{A}}-\E[h'_{\mathrm{A}}]| \!\leq t \mid B_1) \cdot \Pr(B_1) \\
%     & \quad = \Pr(|h'_{\!A}\!- \E[h'_{\!A}|f'_{\mathrm{A}}]+\E[h'_{\!A}|f'_{\mathrm{A}}] -\E[h'_{\!A}]| \leq t|B_1)\! \cdot\! \Pr(B_1) \\
%     & \quad \geq 1-\Pr(|h'_{\mathrm{A}}- \E[h'_{\mathrm{A}}|f'_{\mathrm{A}}]| \geq t_1|B_1) \\
%     & \quad \quad \quad - \Pr(|\E[h'_{\mathrm{A}}|f'_{\mathrm{A}}] -\E[h'_{\mathrm{A}}]| \geq (t-t_1)|B_1) - \Pr(B_1^{c})
% \end{align*}
\begin{align*}
    &\Pr(|h'_{\mathrm{A}}-\E[h'_{\mathrm{A}}]| > t) \\
    & \quad = \Pr(|h'_{\mathrm{A}}- \E[h'_{\mathrm{A}}|f'_{\mathrm{A}}]+\E[h'_{\mathrm{A}}|f'_{\mathrm{A}}] -\E[h'_{\mathrm{A}}]| > t) \\
    & \quad \leq \Pr(|h'_{\mathrm{A}}- \E[h'_{\mathrm{A}}|f'_{\mathrm{A}}]| > t_1 \mid B_1) \\
    & \quad \quad \quad + \Pr(|\E[h'_{\mathrm{A}}|f'_{\mathrm{A}}] -\E[h'_{\mathrm{A}}]| > t_2 \mid B_1) + \Pr(B_1^{c})
\end{align*}
where $t_1+t_2 \leq t$.

Now, we will bound these three terms individually.
First, using Hoeffding's inequality,
$$\Pr(B_1^c) \leq 2 \exp \left(-2 \delta^2 G \bigg(\frac{f_{\mathrm{A}}}{G} \bigg(1-\frac{4p}{3} \bigg)+\frac{p}{3} \bigg)^2 \right) $$
To have $\Pr(B_1^c) \leq 2 \exp(-C_1)$, it suffices to have $2 \delta^2 G \left(\dfrac{f_{\mathrm{A}}}{G} \bigg(1-\dfrac{4p}{3} \bigg)+\dfrac{p}{3} \right)^2 \geq C_1$. Thus, we take  $$\delta := \dfrac{\sqrt{C_1}}{\sqrt{2G}(\frac{f_{\mathrm{A}}}{G}(1-\frac{4p}{3})+\frac{p}{3})}.$$

Note that $\E[h'_{\mathrm{A}}|f'_{\mathrm{A}}]=\frac{NL}{G}(f'_{\mathrm{A}}(1-\frac{4s}{3})+G \frac{s}{3}) $. Also, $\E[h'_{\mathrm{A}}]=\frac{NL}{G}(\E[f'_{\mathrm{A}}](1-\frac{4s}{3})+G\frac{s}{3})$. Thus, conditioned on the event $B_1$, with probability 1,
\begin{align*}
  |\E[h'_{\mathrm{A}}|f'_{\mathrm{A}}]-\E[h'_{\mathrm{A}}]| &\leq \delta NL \frac{\E[f'_{\mathrm{A}}]}{G} \bigg(1-\frac{4s}{3} \bigg) \\
  & = \delta NL \bigg(\frac{f_{\mathrm{A}}}{G} \bigg(1-\frac{4p}{3} \bigg)\!+\frac{p}{3} \bigg) \bigg(1-\frac{4s}{3} \bigg) \\
  & = NL \bigg(1-\frac{4s}{3} \bigg) \sqrt{\frac{C_1}{2G}} \\
  &\leq NL \sqrt{\frac{C_1}{2G}} \ := \ t_2
\end{align*}

Now, we will bound $\Pr(|h'_{\mathrm{A}}- \E[h'_{\mathrm{A}}|f'_{\mathrm{A}}]| \geq t_1 \mid B_1)$. Let r.v. $P_i$ determines the starting position of $i$-th read. Let $W_i$ be random substitution process acting on $i$-th read. The pair $S_i:=(P_i,W_i)$ completely specifies the $i$-th read. Now, $S_1,\cdots,S_N$ are independent, and this independence holds even when conditioned on $f'_{\mathrm{A}}$ or on the event $B_1$. Let $Z=g(S_1,\cdots,S_N)$ denote the total number of occurrences of the symbol $\mathrm{A}$ across all $N$ reads. Since each read has fixed length $L$, changing a single read $S_i$ can change the value of $g$ by atmost $L$, i.e., for each $i=1,\cdots,N, \ |g(s_1,\cdots,s_{i-1},s_{i},s_{i+1},\cdots,s_N)| -  |g(s_1,\cdots,s_{i-1},s'_{i},s_{i+1},\cdots,s_N)| \leq L$. Thus, using McDiarmid's inequality, we get 
\begin{align*}
\Pr(|g(S_1,\cdots,S_N)-\E[g(S_1,\cdots,S_N)]&| \geq t) \\ &\leq 2 \exp \bigg(-\dfrac{2t^2}{NL^2} \bigg).
\end{align*}

Conditioning on the event $B_1$ and $f'_{\mathrm{A}}=q$, we obtain

$$\Pr(|h'_{\mathrm{A}}- \E[h'_{\mathrm{A}}|f'_{\mathrm{A}}=q]| \geq t_1 \mid B_1,f'_{\mathrm{A}}=q) \leq  2\exp \bigg(-\dfrac{2t_{1}^{2}}{NL^2} \bigg).$$

Since this bound holds uniformly for all values of $q$, it follows that
$$\Pr(|h'_{\mathrm{A}}- \E[h'_{\mathrm{A}}|f'_{\mathrm{A}}]| \geq t_1 \mid B_1) \leq  2\exp \bigg(-\dfrac{2t_{1}^{2}}{NL^2} \bigg).$$
To make $\Pr(|h'_{\mathrm{A}}-\E[h'_{\mathrm{A}}|f'_{\mathrm{A}}]|\geq t_1 \mid B_1)\le 2\exp(-C_2)$, we set $t_1 := L\sqrt{NC_2/2}$.

Now, $t_1+t_2 \leq t=|NL-4h_{\mathrm{A}}| \cdot p\epsilon/3$ yields the condition
$$ L\sqrt{\frac{NC_2}{2}}+NL\sqrt{\frac{C_1}{2G}} \leq \frac{|NL-4h_{\mathrm{A}}| \, p\epsilon}{3}, $$
which is equivalent to
$$ \left|\frac{h_{\mathrm{A}}}{NL}-0.25\right| \geq \frac{3}{4p\epsilon} \left( \sqrt{\frac{C_2}{2N}}+\sqrt{\frac{C_1}{2G}} \right) =: u $$

Define the event $B_2:=\left\{\left|\frac{h_{\mathrm{A}}}{NL}-\E\left[\frac{h_{\mathrm{A}}}{NL}\right]\right|\leq u_2\right\}$. Conditioned on $B_2$, a sufficient condition for
$\left|\frac{h_{\mathrm{A}}}{NL}-0.25\right|\geq u$ is
$\left|\E\left[\frac{h_{\mathrm{A}}}{NL}\right]-0.25\right|\geq u+u_2$.  Since $\E\left[\frac{h_{\mathrm{A}}}{NL}\right]
=\frac{f_{\mathrm{A}}}{G}\left(1-\frac{4s}{3}\right)+\frac{s}{3}$, it suffices to have
%the above condition can be written as 
$\left|\frac{f_{\mathrm{A}}}{G}-0.25\right|
\geq u+u_2+\left|\frac{4s}{3}\frac{f_{\mathrm{A}}}{G}-\frac{s}{3}\right|$. But this, if we set $u_2 := \sqrt{C_3/(2N)}$, is precisely the inequality assumed in the statement of the theorem. With this choice of $u_2$, we further have, by McDiarmid’s inequality,
$\Pr(B_2^c)=\Pr(|h_{\mathrm{A}}-\E[h_{\mathrm{A}}]| > u_2 NL)\leq 2\exp(-2u_2^2N) = 2\exp(-C_3)$.
%, it suffices to choose $u_2\geq \sqrt{C_3/(2N)}$.

Putting it all together using the union bound, we obtain
\begin{align*}
    \Pr(|\hat p-p|\leq \epsilon p) & \geq 1- \Pr(|h'_{\mathrm{A}}- \E[h'_{\mathrm{A}}|f'_{\mathrm{A}}]| > t_1 \mid B_1) \\
    & \quad \quad - \Pr(|\E[h'_{\mathrm{A}}|f'_{\mathrm{A}}] -\E[h'_{\mathrm{A}}]| > t_2 \mid B_1) \\
    & \quad \quad - \Pr(B_1^{c}) - \Pr(B_2^{c}) \\
    & \geq 1-2e^{-C_2}-0-2e^{-C_1}-2e^{-C_3} \\
    & = 1-2e^{-C_1}-2e^{-C_2}-2e^{-C_3}.
\end{align*}

\end{proof}

A similar table to Table~\ref{tab:min_deviation_fA} can be obtained for different values of $G,N,p,s,$ and $\epsilon$. As an illustration, consider $G=10^7$, $L=10^3$, $N=10^6$, $p=0.2$, $s=0.03$, and $\epsilon=0.1$. For these parameters, ensuring $\Pr(|\hat p-p|\le \epsilon p)\ge 1-10^{-3}$ requires $\left|\frac{f_{\mathrm{A}}}{G}-0.25\right|\ge 0.112$. 

The main bottleneck of this analysis is that achieving $|\hat p-p|\le \epsilon p$ with high probability requires a large number of reads and a small sequencing error rate. 

\section{Simulations and Results} \label{sec:sim}

We evaluate the performance of all our proposed estimators using both synthetic and real genomic sequences. Synthetic sequences are generated under an i.i.d.\ model, where each symbol is drawn independently according to a fixed probability distribution. Real sequences are extracted from the human T2T-CHM13v2.0 reference genome \cite{nurk2022complete}.

For real data, we consider three sequences. Two of them, \textbf{D-easy} and \textbf{D-hardest}, were previously used in \cite{Wu2025.06.19.660607} and have length $10^5$. In addition, we include a longer sequence of length $10^6$, referred to as \textbf{Chr9-HSat}. 

For the purposes of this discussion, we refer to the estimators from Section~\ref{sec:nonseq} as ``non-sequencing estimators'' and to those from Section~\ref{sec:seq} as ``sequencing-based estimators''. In the case of non-sequencing estimators, where the complete $k$-mer frequency tables of both sequences $x$ and $y$ are available, we compare our estimators for $k=1$ and $k=30$ with the estimator proposed in \cite{Wu2025.06.19.660607}, which we denote by $\widehat{p}_{AH}$.\footnote{The subscript $AH$ is meant to indicate that the estimator uses abundance histogram information.} When considering sequencing-based estimators, where only reads from $x$ and $y$ are available, we evaluate only our estimators for $k=1$ and $k=30$. We use $\widehat{p}_{NS}$ to denote our non-sequencing estimators and $\widehat{p}_{S}$ for the sequencing-based estimators. Unless stated otherwise, the sequencing coverage $c = NL/G$ is fixed to 30, with read length $L = 1000$.

We evaluate the performance of all estimators using box plots of the relative error
$e = \frac{\widehat{p}}{p} - 1$, as a function of the parameters $p$, $s$, $c$, and $G$.

\subsection{Real sequences}

For each real sequence and each value of the mutation rate $p$, we generate 100 mutated sequences. The characteristics of the three sequences used in our experiments are as follows:
\begin{enumerate}
\item \textbf{D-easy:} A length-$10^5$ arbitrarily chosen substring from chr6, with no unusual repeat annotations. For $k=30$, the number of distinct $k$-mers is 99442 and the maximum multiplicity of any $k$-mer is 9. For $k=1$, we use the 1-mer $\mathrm{T}$, for which $f_{\mathrm{T}}/G = 0.301$.
\item \textbf{D-hardest:} A length-$10^6$ highly repetitive subsequence from a region annotated as ``Active $\alpha$Sat HOR'' in the chr21 centromere. For $k=30$, the number of distinct $k$-mers is 3987 and the maximum multiplicity is 127. For $k=1$, we use the 1-mer $\mathrm{G}$, with $f_{\mathrm{G}}/G = 0.168$.
\item \textbf{Chr9-HSat:} A length-$10^6$ repetitive substring from a pericentromeric region of chr9 annotated as satellite DNA (HSAT2/HSAT3). For $k=30$, the number of distinct $k$-mers is 58244 and the maximum multiplicity is 1250. For $k=1$, we use the 1-mer $\mathrm{C}$, with $f_{\mathrm{C}}/G = 0.169$.
\end{enumerate}

Figures~\ref{fig1}(\subref{fig_1a}), \ref{fig1}(\subref{fig_1b}), and \ref{fig1}(\subref{fig_1e}) compare the estimators in the non-sequencing case. For $k=30$, our estimator $\widehat{p}_{NS}$ (given by~\eqref{def:phat-4}) significantly outperforms $\widehat{p}_{AH}$ for the D-hardest and Chr9-HSat sequences. This improvement arises because $\widehat{p}_{NS}$ accounts for $k$-mer multiplicities in the mutated sequence as well. For D-easy sequence, where most $k$-mers occur only once, the performance of $\widehat{p}_{NS}$ and $\widehat{p}_{AH}$ is similar.

We also observe that the estimators using $k=30$ become unstable at higher mutation rates ($p \geq 0.3$), whereas the estimator that uses $k=1$, given by \eqref{def:phat-1}, performs poorly at low mutation rates but improves as $p$ increases. The Chr9-HSat sequence provides an example where both the $k=30$ estimators --- $\widehat{p}_{NS}$ given by~\eqref{def:phat-4} as well as $\widehat{p}_{AH}$ --- perform poorly even at low mutation rates, while the $k=1$ estimator $\widehat{p}_{NS}$ given by~\eqref{def:phat-1} achieves relatively better performance.

Figures~\ref{fig1}(\subref{fig_1c}), \ref{fig1}(\subref{fig_1d}), and \ref{fig1}(\subref{fig_1e}) present results for the sequencing-based estimators. For D-easy and D-hardest, trends similar to the non-sequencing case are observed: the $k=1$ estimator given by~\eqref{def:phat-5} performs better at higher mutation rates, whereas the $k=30$ estimator given by~\eqref{def:phat-8} performs better at lower mutation rates. Going from non-sequencing to sequencing-based estimators leads to a slight degradation in performance, but the overall performance remains within acceptable limits.

For the Chr9-HSat sequence, we only report results for the $k=1$ sequencing-based estimator $\widehat{p}_S$ given by~\eqref{def:phat-5}; for $k=30$ even the non-sequencing estimators perform poorly. In this case, $\widehat{p}_{S}$ performs well due to the large sequence length and the significant deviation of $f_{\mathrm{C}}/G$ from $0.25$.

Figures~\ref{fig1}(\subref{fig_1f}) and \ref{fig1}(\subref{fig_1g}) illustrate the effect of varying sequencing coverage $c$ and sequencing error rate $s$ while fixing $p = 0.05$, for the D-easy and D-hardest sequences. The estimator $\widehat{p}_{S}$ given by~\eqref{def:phat-5} is largely insensitive to changes in $s$, whereas the performance of $\widehat{p}_{S}$ given by~\eqref{def:phat-8} degrades as $s$ increases. This degradation affects D-easy more than D-hardest since most of the $k$-mers in D-easy occur only once. Therefore, as $s$ increases, most of the $k$-mers don't appear in the reads, allowing spurious $k$-mers generated by sequencing errors to dominate.

\subsection{Synthetic sequences}

For randomly generated synthetic sequences, we fix $p = s = 0.05$ and vary $f_{\mathrm{A}}/G$, $c$, and $G$. We consider sequence lengths $G = 10^5$ and $10^6$. For each value of $f_{\mathrm{A}}/G$, we generate 20 independent reference sequences and simulate 5 mutated sequences from each reference.

Figure~\ref{fig2}(\subref{fig_2a}) shows the performance of the $k=1$ sequencing-based estimator $\widehat{p}_{S}$ given by~\eqref{def:phat-5} as a function of $f_{\mathrm{A}}/G$, $c$, and $G$. The estimator accuracy improves as $|f_{\mathrm{A}}/G - 0.25|$ increases. Increasing the sequence length $G$ leads to a substantial improvement in performance, while increasing the coverage $c$ only slightly improves the performance. The estimator becomes unstable when $|f_{\mathrm{A}}/G - 0.25|$ is small.

Fig.~\ref{fig2}(\subref{fig_2b}) evaluates the performance of $\widehat{p}_{S}$ given by~\eqref{def:phat-8}. In contrast to the $k=1$ estimator, this estimator performs well even when $f_{\mathrm{A}}/G = 0.25$. Here, performance improves slightly with increasing $G$ and more significantly with increasing $c$.

\begin{figure*}[t]
    \centering

    \begin{subfigure}{0.4\textwidth}
        \includegraphics[width=\linewidth]{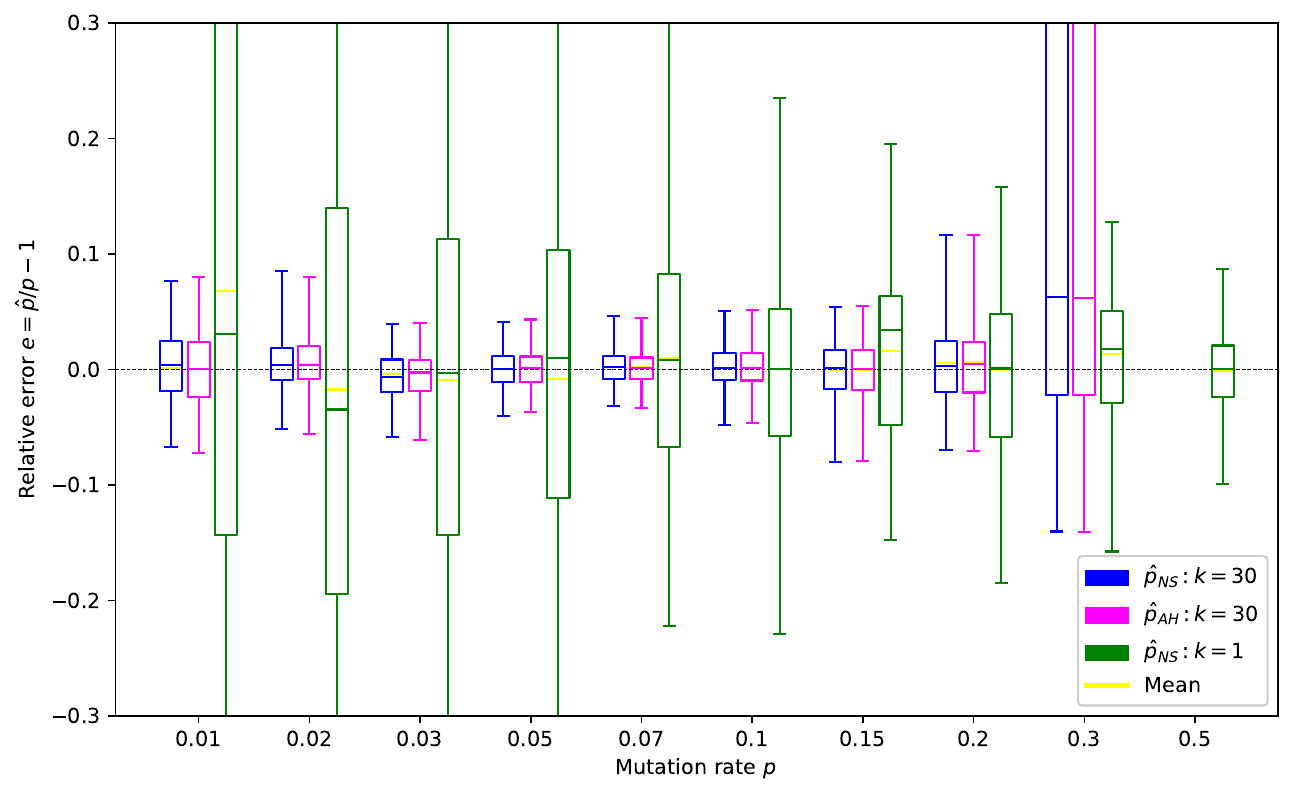}
        \caption{D-easy non-sequencing}
        \label{fig_1a}
    \end{subfigure}\hfill
    \begin{subfigure}{0.4\textwidth}
        \includegraphics[width=\linewidth]{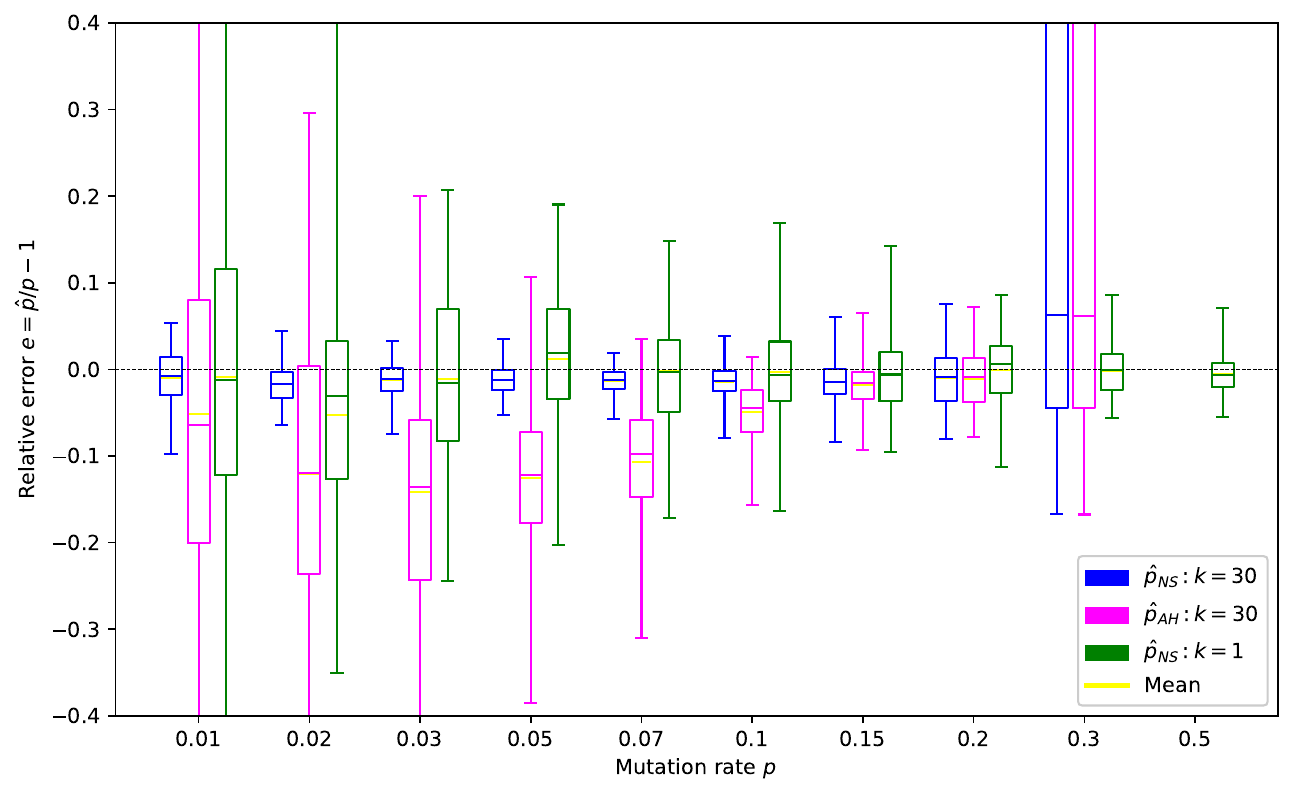}
        \caption{D-hardest non-sequencing}
         \label{fig_1b}
    \end{subfigure}
    
    \vspace{1mm}
    
    \begin{subfigure}{0.4\textwidth}
        \includegraphics[width=\linewidth]{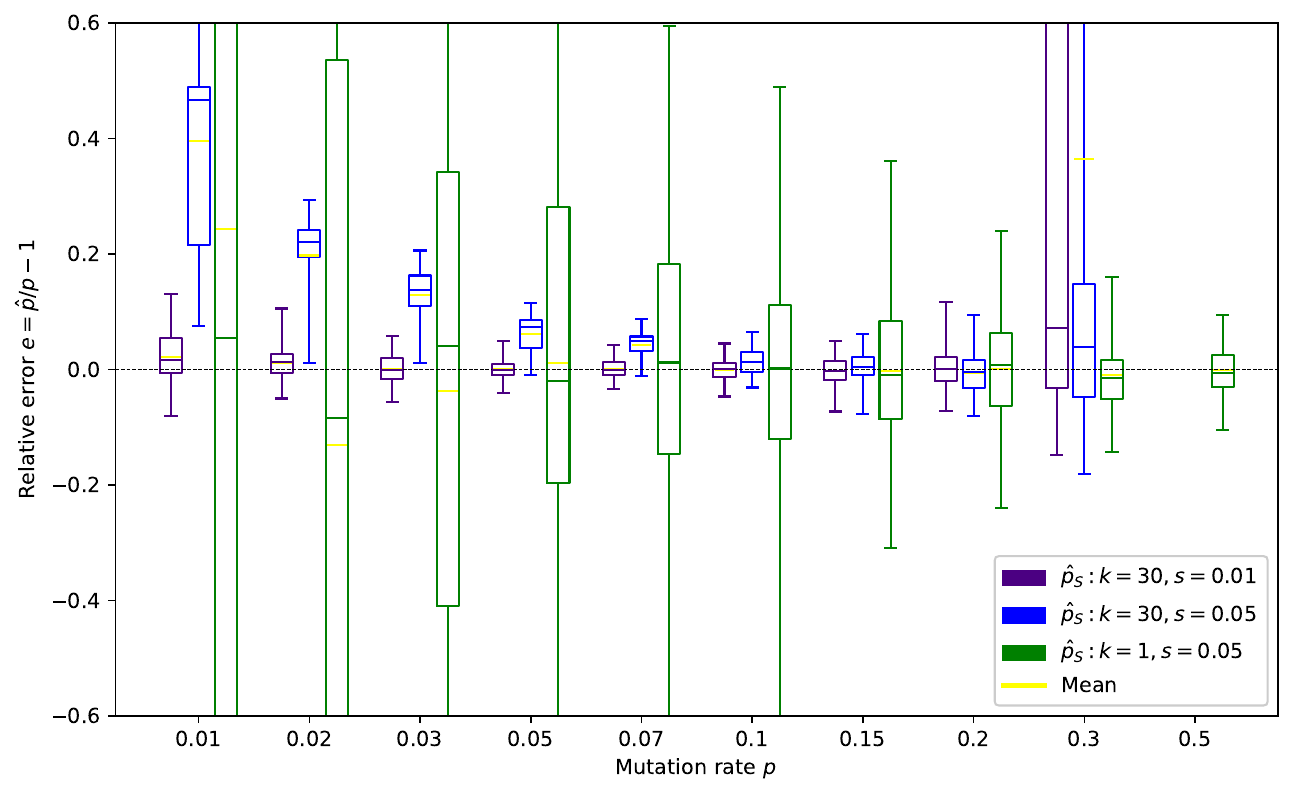}
        \caption{D-easy sequencing}
         \label{fig_1c}
    \end{subfigure}\hfill
    \begin{subfigure}{0.4\textwidth}
        \includegraphics[width=\linewidth]{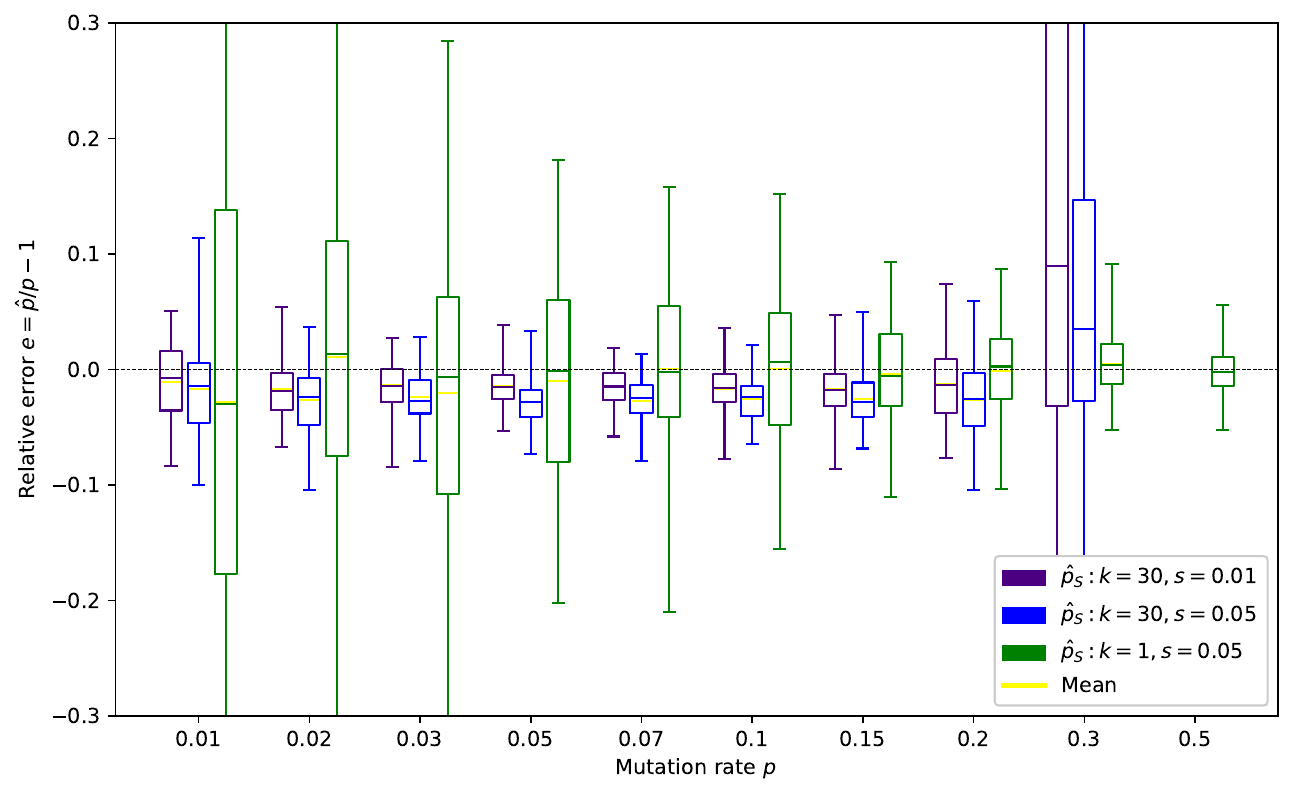}
        \caption{D-hardest sequencing}
         \label{fig_1d}
    \end{subfigure}

    \vspace{1mm}
    
    \begin{subfigure}{0.35\textwidth}
        \includegraphics[width=\linewidth]{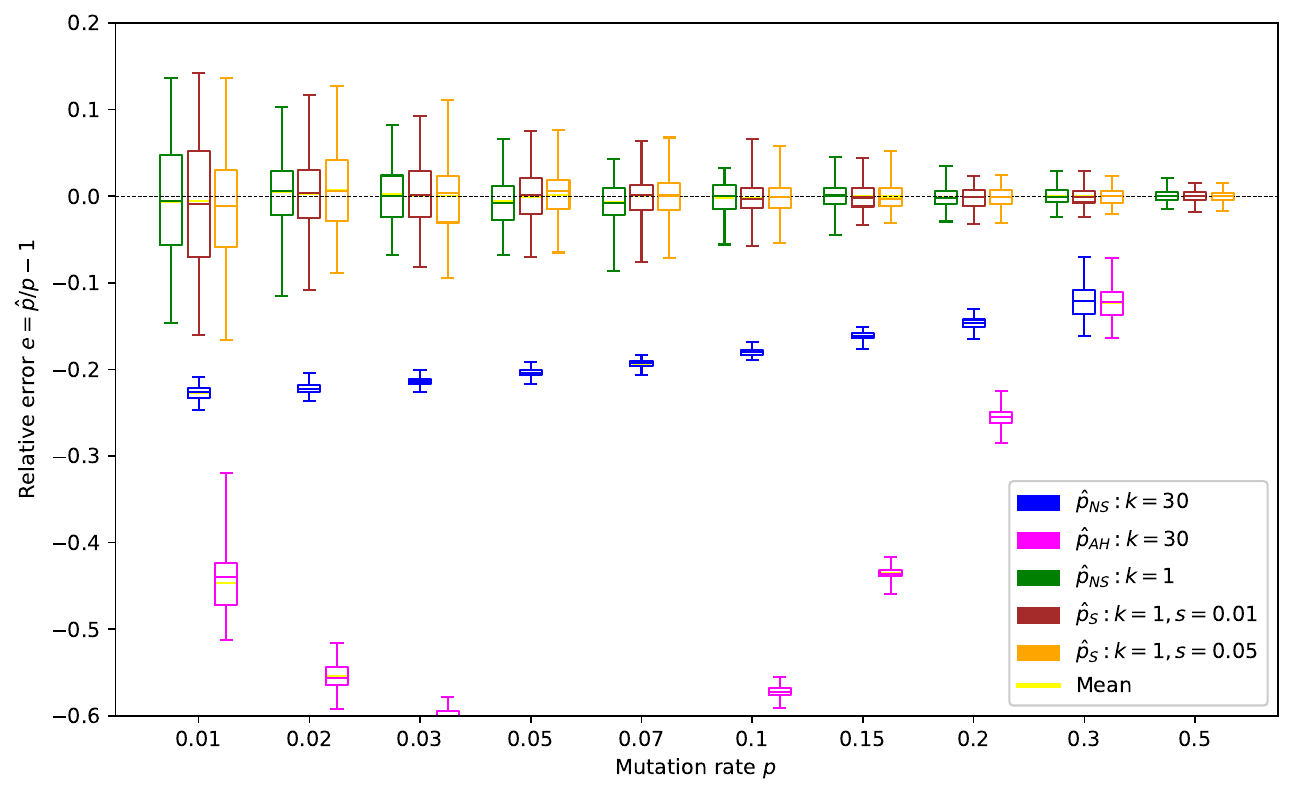}
        \caption{Chr9-HSat}
         \label{fig_1e}
    \end{subfigure}\hfill
    \begin{subfigure}{0.25\textwidth}
        \includegraphics[width=\linewidth]{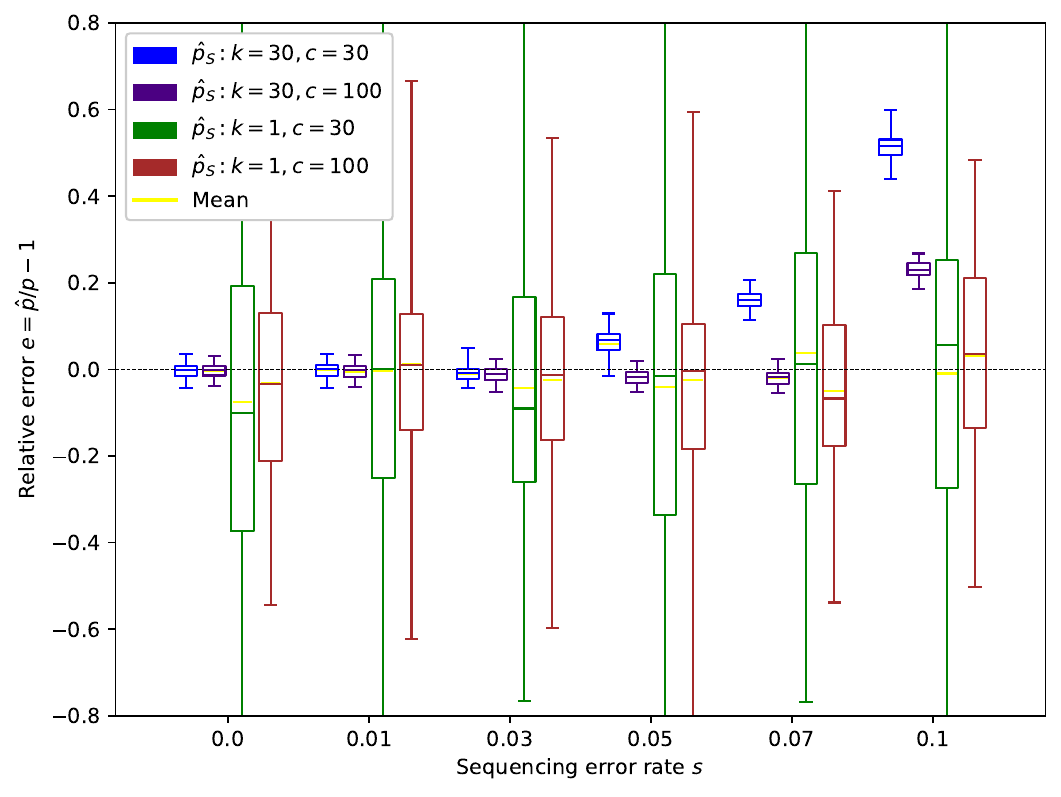}
        \caption{D-easy: Fixed $p=0.05$, varying $s$ and $c$}
         \label{fig_1f}
    \end{subfigure}\hfill
    \begin{subfigure}{0.25\textwidth}
        \includegraphics[width=\linewidth]{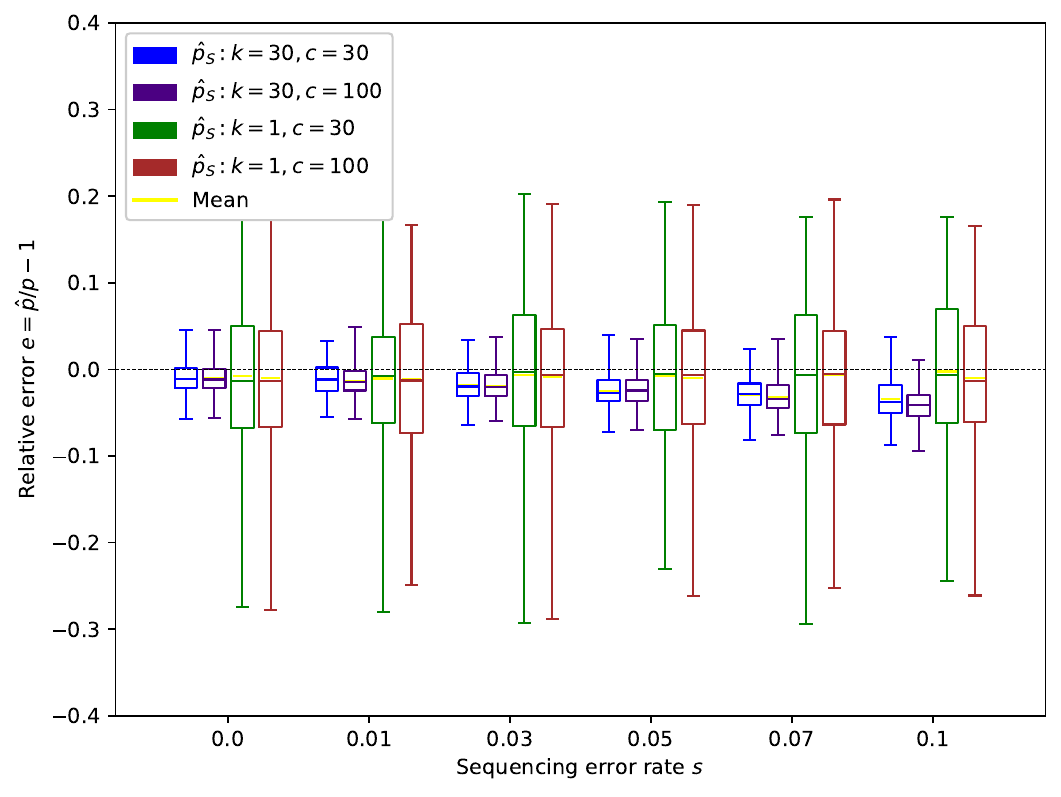}
        \caption{D-hardest: Fixed $p=0.05$, varying $s$ and $c$}
         \label{fig_1g}
    \end{subfigure}

    \caption{Real genomic sequences}
    \label{fig1}
\end{figure*}

\begin{figure*}[b]
    \centering

    \begin{subfigure}{0.6\textwidth}
        \includegraphics[width=\linewidth]{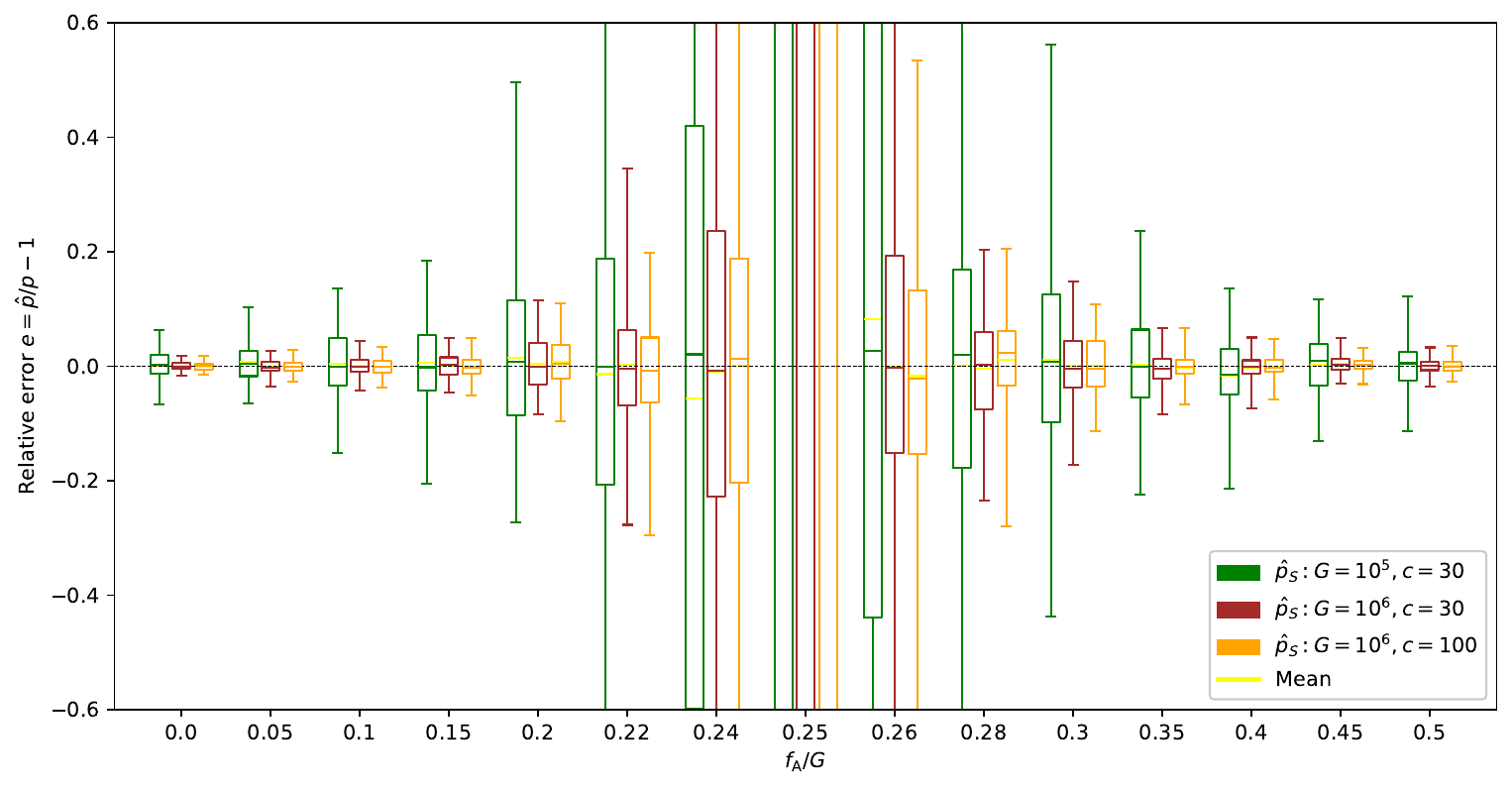}
        \caption{$k=1$}
        \label{fig_2a}
    \end{subfigure}\hfill
    \begin{subfigure}{0.2\textwidth}
        \includegraphics[width=\linewidth]{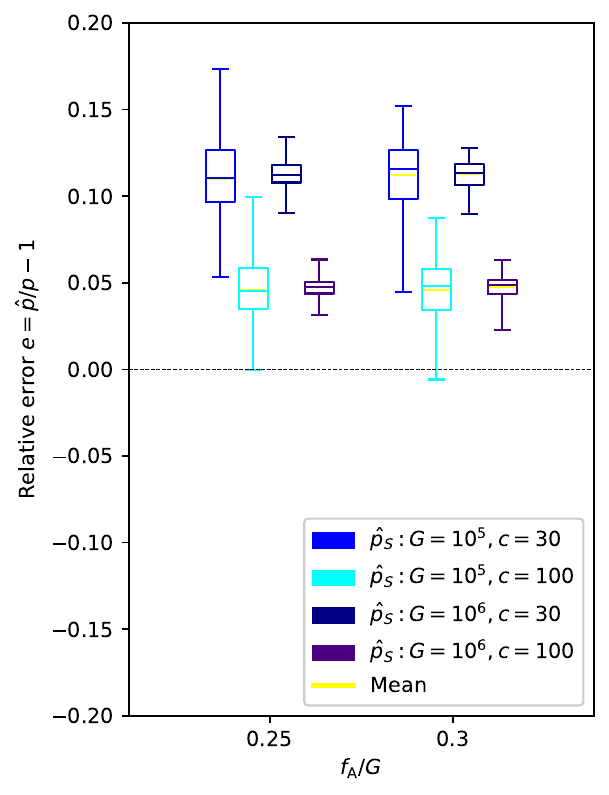}
        \caption{$k=30$}
        \label{fig_2b}
    \end{subfigure}

    \caption{Randomly generated synthetic sequences}
    \label{fig2}
\end{figure*}

\section{Conclusion and Future Work} \label{sec:conc}

In this work, we extend estimation of substitution rate between two sequences to the practically relevant setting in which only unassembled sequencing reads are available, and the underlying sequences are unknown. Our approach relies on $k$-mer statistics, and we propose two kinds of estimators: one based on $k=1$ and another based on large $k$, for which we use $k=30$ in practice. We also derived theoretical performance guarantees for the $k=1$ estimator.

Both the estimators have their pros and cons. The $k=1$ estimator is unbiased and performs well whenever $|f_v/G - 0.25| > 0$ for at least one $v \in \{A,G,C,T\}$. Its accuracy improves with increasing sequence length $G$ and higher mutation rates $p$, and it is relatively robust to sequencing errors. Moreover, the $k=1$ estimator has substantially lower time and space complexity and scales better with sequence length. In contrast, the estimator with $k=30$ is slightly biased, with a bias that depends on the underlying sequence. Nevertheless, it performs better in practice for less repetitive sequences and in regimes where the probability that a $k$-mer mutates to another $k$-mer in the sequence is low, particularly at small mutation rates. However, this estimator is more sensitive to sequencing errors.

A direction for future work is to establish theoretical performance guarantees for the large-$k$ estimator. In addition, there exist regimes and sequences for which both estimators become unstable or fail to provide accurate estimates. This motivates the open problem of designing an estimator that performs well across different regimes and sequence types, or at least for most real genomic sequences. Another possible future direction is to extend this framework to account for other types of mutations and sequencing errors such as insertions and deletions.

\textbf{Acknowledgements.}\ The authors would like to acknowledge useful discussions with Daanish Mahajan, Chirag Jain, and Paul Medvedev.

\clearpage
%newpage
\raggedbottom
\bibliographystyle{ieeetr}
\bibliography{references}

@article{Ondov029827,
  title={Mash: Fast genome and metagenome distance estimation using {MinHash}},
  author={Ondov, B. D. and Treangen, T. J. and Melsted, P. and Mallonee, A. B. and Bergman, N. H. and Koren, S. and Phillippy, A. M.},
  journal={Genome Biology},
  volume={17},
  number={1},
  year={2016},
  pages={132},
  doi={10.1186/s13059-016-0997-x},
  url={https://doi.org/10.1186/s13059-016-0997-x}
}

@article{Blanca2021.01.15.426881,
  title={The statistics of k-mers from a sequence undergoing a simple mutation process without spurious matches},
  author={Blanca, Antonio and Harris, Robert S and Koslicki, David and Medvedev, Paul},
  journal={Journal of Computational Biology},
  volume={29},
  number={2},
  pages={155--168},
  year={2022},
  publisher={Mary Ann Liebert, Inc., publishers 140 Huguenot Street, 3rd Floor New~…}
}

@misc{Wu2025.06.19.660607,
  title={A k-mer-based estimator of the substitution rate between repetitive sequences},
  author={Wu, Haonan and Blanca, Antonio and Medvedev, Paul},
  note={bioRxiv e-print 2025.06.19.660607},
  year={2025},
  doi = {10.1101/2025.06.19.660607},
  url={https://doi.org/10.1101/2025.06.19.660607}
}

@misc{Hera2025.05.14.653858,
  title={Estimation of substitution and indel rates via k-mer statistics},
  author={Hera, Mahmudur Rahman and Medvedev, Paul and Koslicki, David and Blanca, Antonio},
  note={bioRxiv e-print 2025.05.14.653858},
  year={2025},
  doi = {10.1101/2025.05.14.653858},
  url={https://doi.org/10.1101/2025.05.14.653858}
}

@incollection{mcdiarmid1989method, 
    place={Cambridge}, 
    series={London Mathematical Society Lecture Note Series}, 
    title={On the method of bounded differences}, 
    booktitle={Surveys in Combinatorics, 1989: Invited Papers at the Twelfth British Combinatorial Conference}, 
    publisher={Cambridge University Press}, 
    author={McDiarmid, Colin}, 
    editor={Siemons, J.Editor}, 
    year={1989}, 
    pages={148--188}, 
    collection={London Mathematical Society Lecture Note Series}
}

@article{hoeffding1963probability,
  title={Probability inequalities for sums of bounded random variables},
  author={Hoeffding, Wassily},
  journal={Journal of the American Statistical Association},
  volume={58},
  number={301},
  pages={13--30},
  year={1963},
  publisher={Taylor \& Francis}
}

@article{nurk2022complete,
  title={The complete sequence of a human genome},
  author={Nurk, Sergey and Koren, Sergey and Rhie, Arang and Rautiainen, Mikko and Bzikadze, Andrey V and Mikheenko, Alla and Vollger, Mitchell R and Altemose, Nicolas and Uralsky, Lev and Gershman, Ariel and others},
  journal={Science},
  volume={376},
  number={6588},
  pages={44--53},
  year={2022},
  publisher={American Association for the Advancement of Science}
}
\end{document}